\documentclass[%
 twocolumn,
 10pt,
 superscriptaddress,
 longbibliography,
 amsmath,amssymb,
 aps,
 pra,
floatfix,
]{revtex4-1}

\usepackage{graphicx,amsthm}
\usepackage{dcolumn}
\usepackage{bm}
\usepackage{times}
\usepackage{amsmath}
\usepackage{enumerate}
\usepackage[caption=false]{subfig}
\usepackage{qcircuit}
\usepackage{color}
\usepackage[colorlinks=true,linkcolor=blue,citecolor=red, linktocpage=true,breaklinks=true]{hyperref}

\usepackage{epstopdf}

\newcommand{\eq}{\begin{equation}}
\newcommand{\en}{\end{equation}}
\newcommand{\eqa}{\begin{eqnarray}}
\newcommand{\ena}{\end{eqnarray}}

\newtheorem{theo}{Theorem}

\newcolumntype{Y}{>{\centering\arraybackslash}X}
\begin{document}

\title{Depth optimization of quantum search algorithms beyond Grover's algorithm}

\author{Kun \surname{Zhang}}
\affiliation{Department of Chemistry, State University of New York at Stony Brook, Stony Brook, New York 11794-3400, USA}
\author{Vladimir  E. \surname{Korepin}}
\affiliation{C.N. Yang Institute for Theoretical Physics, State University of New York at Stony Brook, Stony Brook, New York 11794-3840, USA}
\affiliation{
Institute for Advanced Computational Science, State University of New York at Stony Brook, Stony Brook, New York 11794-5250, USA}

\date{\today}

\begin{abstract}
	
	Grover's quantum search algorithm provides a quadratic speedup over the classical one. The computational complexity is based on the number of queries to the oracle. However, depth is a more modern metric for noisy intermediate-scale quantum computers. We propose a new depth optimization method for quantum search algorithms. We show that Grover's algorithm is not optimal in depth. We propose a quantum search algorithm, which can be divided into several stages. Each stage has a new initialization, which is a rescaling of the database. This decreases errors. The multistage design is natural for parallel running of the quantum search algorithm. 
	
\end{abstract}

\maketitle

\section{\label{sec:intro}Introduction}

Quantum algorithms are designed to outperform the best classical ones \cite{NC10}. Many nondeterministic NP-hard problems still have only the exhaustive search way to solve them \cite{BBBV96}. The one-way function (oracle) $f(x)$ ($f:\{0,1\}^n\rightarrow \{0,1\}$) can identify the solution state: if $t$ is the solution (target state), then $f(t)=1$; otherwise the one-way function output is zero. The classical way to execute the exhaustive search is by querying each state in the database (of $N$ items) by the one-way function. In the worst case, the total number of queries to the oracle is $N-1$. The principle of quantum superposition provides a superior way to perform the exhaustive search. Suppose that $N=2^n$, where $n$ is the number of qubits to represent the database. Grover's algorithm can find one target state with oracle complexity $\mathcal O(\sqrt N)$, which quadratically outperforms the classical algorithm \cite{Grover97,GK17}. The oracle in Grover's algorithm is $U_f$: $U_f|x\rangle|y\rangle=|x\rangle|f(x)\oplus y\rangle$ with $x\in\{0,1\}^n$ and $y\in\{0,1\}$.

Quantum computers have been vastly developed over the last ten years \cite{Barends14,BHLSL16,FMLLDM17,Arute19}. Still shallow-depth algorithms can be realized on real quantum computers (for the noisy intermediate-scale quantum (NISQ) era, see Ref. \cite{Preskill18}). The width (the number of physical qubits) represents the size of quantum computers. The algorithm's depth (the number of consecutive  gate operations) represents the physical implementation time for the algorithm. Multiplying the width and depth we get the quantum volume, which gives a metric for NISQ computers \cite{CBSNG19}. Coherence time is limited in NISQ computers. A set of gates which can approximate any unitary operation is called the universal quantum gate set (Solovay–Kitaev theorem) \cite{NC10}. We assume that the quantum computer is equipped with a universal quantum gate set. So, the depth is counted by universal quantum gate operations. 

The quantum oracle $U_f$ is realized by quantum gates from the universal quantum gate set. We assume that the depth of the quantum oracle scales polynomially with $n$ \cite{FMLLDM17}. The oracle complexity would be equivalent to the depth complexity if the quantum oracle would be the only operation realized in Grover's algorithm. However, it is not true. Another unitary operation (diffusion operator) is required for Grover's algorithm \cite{Grover97,GK17}. How to choose the diffusion operator is related to the initial state preparation \cite{Grover98,BHMT00}. The unstructured population space $\{0,1\}^n$ (database) can be prepared in an equal superposition state on a quantum computer polynomial efficiently:
\begin{equation}
\label{def s n}
    |s_n\rangle=H^{\otimes n}|0\rangle^{\otimes n}
\end{equation}
with single-qubit Hadamard gate $H$ \cite{NC10}. Note that the initial state $|s_n\rangle$ can be efficiently prepared with a depth of one circuit. The diffusion operator has the constraint that the state $|s_n\rangle$ is the eigenvector of the diffusion operator with eigenvalue 1 \cite{Tulsi12,Tulsi15}. 

Grover's algorithm is the only threat to postquantum cryptography. The postquantum cryptography standardization proposed by NIST in 2016 introduced the depth bound. Recently, more studies focused on the resource estimation, such as width and depth, for Grover's algorithm instead of the traditional oracle complexity \cite{KHJ18,JNRV19}. Grover's algorithm is optimal in oracle complexity \cite{BBHT98,Zalka99}. However, no research addressed the depth of the quantum search algorithm. Surprisingly, the depth of the diffusion operator can be reduced to one \cite{Kato05,JRW17}. However, these algorithms have 1/2 maximal successful probability, and the expected depth is not as efficient as the original Grover's algorithm. Inspired by the quantum partial search algorithm (QPSA) \cite{GR05,KG06,Korepin05,KL06}, we introduce a new depth optimization for the quantum search algorithm. Our algorithm can have lower depth than Grover's algorithm. To further lower the depth, we can apply a divide-and-conquer strategy (combined with depth optimization). The divide-and-conquer strategy means that the search algorithm is realized by several stages. Each stage can find a partial address of the target state. The next-stage initial state is the rescaled version of the last-stage initial state. The divide-and-conquer strategy naturally allows the parallel running of the quantum search algorithm. 

If the oracle takes much more depths than diffusion operator depth, then the oracle complexity will be approximately equivalent to the depth complexity. We can define the ratio between oracle depth and diffusion operator depth. Above a critical ratio, Grover's algorithm is optimal in depth. Based on the depth optimization method proposed in this paper, we show that the critical ratio is $\mathcal O(n^{-1}2^{n/2})$. If we divide the algorithm into two stages, the critical ratio is a constant. 

The paper is organized as follow. In Sec. \ref{sec:Grover}, we briefly review  quantum search algorithms. The first one is Grover's original algorithm and the other is QPSA. We also set up  notations. In Sec. \ref{sec:dep_opt}, we introduce the depth optimization method for the quantum search algorithm. We also show how to combine the divide-and-conquer strategy with depth optimization. In Sec. \ref{sec:alpha}, we talk about the critical ratios. Below the critical ratio, we can have a search algorithm which has lower depth compared to Grover's algorithm. Parallel running of the quantum search algorithm is briefly discussed in Sec. \ref{sec:parallel}. Section \ref{sec:conclusion} gives conclusions and outlook. We wrote three Appendixes. Appendix \ref{Appendix examples 6} provides detailed examples of the $n=6$ search algorithm with depth optimizations; Appendix \ref{Appendix opt examples} lists the numerical details provided in the main text; Appendix \ref{Appendix alpha} shows the numerical values of critical ratios. 

\section{\label{sec:Grover}Review of Quantum Search Algorithms}

\subsection{\label{subsec:Grover}Grover's Algorithm} 

The quantum oracle $U_f$ flips the ancillary qubit, if the target state $|t\rangle$ is fed in. The ancillary qubit can be prepared in the superposition state $H|1\rangle = (|0\rangle-|1\rangle)/\sqrt 2$. Then the oracle gives a sign flip acting on the target state:
\begin{equation}
U_f(1\!\!1_{2^n}\otimes H)|x\rangle\otimes|1\rangle = (-1)^{f(x)}(1\!\!1_{2^n}\otimes H)|x\rangle\otimes|1\rangle
\end{equation}
Here $1\!\!1_{2^n}$ is the identity operator on the $2^n$ dimensional Hilbert space. For convenience, we denote the oracle $U_f$ as
\begin{equation}
\label{def U t}
U_t=1\!\!1_{2^n}-2|t\rangle\langle t|  
\end{equation}
if the ancillary qubit $H|1\rangle$ is prepared. The general phase flip can be constructed as follows: $U_{t,\phi}=1\!\!1_{2^n}-(1-e^{-i\phi})|t\rangle\langle t|$ with complex unit $i=\sqrt{-1}$. The generalized oracle $U_{t,\phi}$ has applications in the sure success search algorithm \cite{BHMT00,MTB18} and the fixed point search algorithm (for an unknown number of target states) \cite{YLC14}. Note that the operator $U_{t,\phi}$ ($\phi\neq \pi$) can be realized by two quantum oracles $U_f$ \cite{YLC14}. In this paper, we do \textit{not} consider the generalized oracle $U_{t,\phi}$ (low depth consideration). We concentrate on the one-target-state case. The depth optimization method in Sec. \ref{sec:dep_opt} can be easily generalized to multitarget cases.

The oracle $U_t$ reflects the state over the plane perpendicular to the target state. The most efficient diffusion operator (unstructured database search) is
\begin{equation}
\label{def I n}
D_n = 2|s_n\rangle\langle s_n|-1\!\!1_{2^n}
\end{equation}
Note that $|s_n\rangle$ defined in Eq. (\ref{def s n}) is the equal superposition of all items in the  database. The operator $D_n$ can be viewed as a reflection of the amplitude in the average. The diffusion operator $D_n$ does not query the oracle. Therefore, the oracle complexity does not include the resource cost by $D_n$. The diffusion operator $D_n$ is single-qubit-gate-equivalent to the generalized $n$-qubit Toffoli gate $\Lambda_{n-1}(X)$ \cite{NC10}. Here $X$ is the NOT gate (Pauli-X gate). The notation $\Lambda_{n-1}(X)$ implies the $n-1$ control qubits NOT gate. When $n=3$, $\Lambda_{2}(X)$ is the Toffoli gate. When $n=2$, $\Lambda_{1}(X)$ is the controlled-NOT (CNOT) gate. How to realize the $\Lambda_{n-1}(X)$ gate on a real quantum computer is highly nontrivial. It is well known that an $n$-qubit $\Lambda_{n-1}(X)$ gate can be constructed with linear $n$ depth or quadratic $n^2$ depth from the universal gate set (CNOT gate plus single-qubit gates) \cite{BBCDMSSSW95}. Recent works also show that the $n$-qubit $\Lambda_{n-1}(X)$ gate can be realized in $\log n$ depth if $n$-qubit ancillary qubits are provided \cite{HLZWW17} or qutrit states are applied \cite{GBDBRC19}. 

One query to oracle $U_t$ defined in Eq. (\ref{def U t}) combined with the diffusion operator $D_n$ defined in Eq. (\ref{def I n}) is called the Grover iteration or Grover operator:
\begin{equation}
\label{def G n}
G_n =  D_n U_t
\end{equation}
See Fig. \ref{fig G n} for the quantum circuit diagram of $G_n$. The diffusion operator $D_n$ reflects the average of the whole database. The operator $G_n$ is also called the {\it global Grover iteration} ({\it global Grover operator}). One Grover operator $G_n$ uses one query to oracle $U_f$. Applying $G_n$ iteratively on the initial state $|s_n\rangle$, the amplitude of the target state will be amplified. After $j$ Grover iterations, the success probability $P_n(j)$ is
\begin{equation}
\label{def P n}
P_n(j) = |\langle t|G_n^j|s_n\rangle|^2 = \sin^2((2j+1)\theta)  
\end{equation}
with $\sin\theta=1/\sqrt N$. When $j$ reaches $j_\text{max}=\lfloor \pi\sqrt N/4\rfloor$, the probability of finding the target state approaches 1. The maximal iteration number $j_\text{max}$ is the square root of $N$. Clearly, Grover's algorithm provides a quadratic speedup compared with the classical algorithm (in oracle complexity). The idea behind Grover's algorithm can be generalized into the amplitude amplification algorithm \cite{BHMT00}.

The success probability (finding the target state) does not scale linearly with the number of iterations. It suggests that Grover's algorithm becomes less efficient when $j$ approaches $j_\text{max}$. Previous works argued that the expected number of iterations $j/P_n(j)$ has the minimum at $j_\text{exp}=\lfloor 0.583\sqrt N \rfloor$, which is smaller than $j_\text{max}$ \cite{BBHT98,GWC00}. When $j$ is $j_\text{exp}$, the success probability is around 0.845. In practice, the iteration number $j_\text{exp}$ has a high probability to find the target state. The measurement result can be verified in classical ways. If the result fails, one has to run the algorithm again. The expected number of oracles is minimized at $j_\text{exp}$.

\begin{figure}
\subfloat[\label{fig G n} $G_n = D_nU_t$.]{%
  \includegraphics[width=0.48\columnwidth]{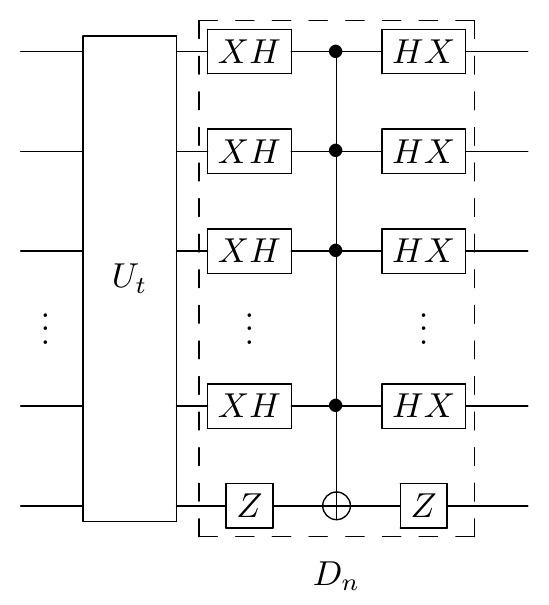}%
}\hfill
\subfloat[\label{fig G m} $G_m = D_mU_t$.]{%
  \includegraphics[width=0.5\columnwidth]{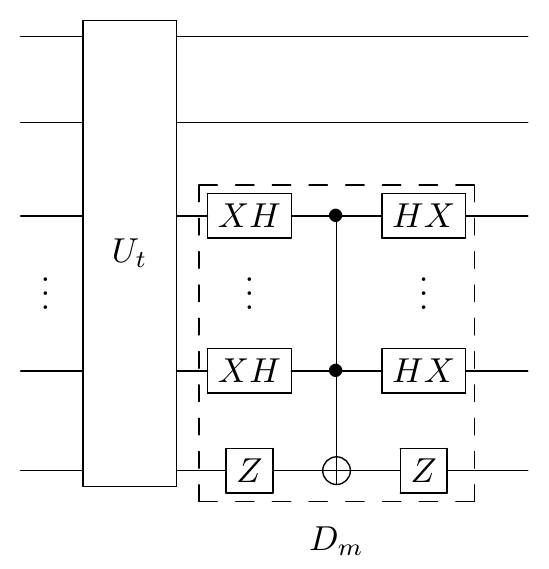}%
}
\caption{Quantum circuits of global Grover operator $G_n$ defined in Eq. (\ref{def G n}) and local Grover operator defined in Eq. (\ref{def G m}). The diffusion operator $D_n$ ($D_m$) is single-qubit-gate equivalent to the $n$-qubit Toffoli gate $\Lambda_{n-1}(X)$ ($m$-qubit Toffoli gate $\Lambda_{m-1}(X)$) \cite{NC10}. Here $X$ and $Z$ are Pauli gates, and $H$ is the Hadamard gate. The subspace where $D_m$ acts can be chosen arbitrarily.}
\label{fig G n m}
\end{figure}

\subsection{\label{subsec:QPSA}Quantum Partial Search Algorithm} 

The QPSA was introduced by Grover and Radhakrishnan \cite{GR05}. Since Grover's algorithm is optimal (in oracle complexity), the QPSA trades accuracy for speed. A database of $N$ items is divided into $K$ blocks: $N=bK$. Here $b$ is the number of items in each block. We can assume that the number $b$ is also a power of 2: $b=2^m$. And the number of blocks is $K=2^{n-m}$. The QPSA can find the block which has the target state. In other words, the QPSA finds the partial $(n-m)$-bit of the target state (which is $n$ bits long). The optimized QPSA can win over Grover's algorithm a number scaling as $\sqrt b$ \cite{GR05,Korepin05,KG06}. A larger block size (less accuracy) gives a faster algorithm. 

Suppose that the address of the target state $|t\rangle$ is divided into $|t\rangle=|t_1\rangle\otimes|t_2\rangle$. Here $t_1$ is $(n-m)$ bits long and $t_2$ is $m$ bits long. The task is to find $t_1$ instead of the whole $t$. Besides the diffusion operator $D_n$ in Eq. (\ref{def I n}), the QPSA introduces a new diffusion operator $D_{n,m}$:
\begin{equation}
    \label{def I n m}
    D_{n,m} = 1\!\!1_{2^{n-m}}\otimes (2|s_m\rangle\langle s_m|-1\!\!1_{2^m})
\end{equation}
The diffusion operator $D_{n,m}$ reflects around the average in a block (simultaneously in each block). The diffusion operator $D_{n,m}$ can be viewed as the rescaled version of $D_n$ in Eq. (\ref{def I n}): the database with size $2^n$ is rescaled into size $2^{m}$. We can define a new Grover operator as
\begin{equation}
    \label{def G m}
    G_{n,m}=D_{n,m}U_t
\end{equation}
See Fig. \ref{fig G m} for the quantum circuit diagram of $G_{n,m}$. The diffusion operator $D_{n,m}$ reflects the average of block items. The operator $G_{n,m}$ is also called the {\it local Grover iteration} ({\it local Grover operator}). For simplicity, we shorten the notations to $D_{m}\equiv D_{n,m}$ and $G_{m}\equiv G_{n,m}$ in the rest of paper. 

The QPSA is realized by  applying operators $G_{m}$ and $G_{n}$ on the initial state $|s_n\rangle$. Then partial bits $t_1$ can be found with high probability (computational basis measurement on the final state). In the QPSA, the amplitudes of all nontarget items in the target block are the same, and the amplitudes of all items in the nontarget blocks are the same. Therefore, we can follow only three amplitudes. Let us introduce a basis:
\begin{subequations}
	\begin{align}
	\label{basis t} &|t\rangle = |t_1\rangle\otimes |t_2\rangle, \\
	\label{basis ntt} &|ntt\rangle = \frac 1 {\sqrt{b-1}}\sum_{j\neq t_2}|t_1\rangle\otimes |j\rangle, \\
	\label{basis u} &|u\rangle =  \frac{1}{\sqrt{N-b}}\left(\sqrt N|s_n\rangle-|t\rangle-\sqrt{b-1}|ntt\rangle\right)
	\end{align}
\end{subequations}
The state $|ntt\rangle$ is the normalized sum of all nontarget states in the target block. The state $|u\rangle$ is the normalized sum of all items in the nontarget blocks. At the new basis, the initial state $|s_n\rangle$ in Eq. (\ref{def s n}) can be rewritten as
\begin{equation}
\label{def s n rewrite}
    |s_n\rangle = \sin \gamma \sin\theta_2|t\rangle + \sin \gamma \cos\theta_2|ntt\rangle+\cos\gamma |u\rangle
\end{equation}
The angle $\theta_2$ is defined as $\sin\theta_2 = 1/\sqrt b$. The angle $\gamma$ is defined as $\sin\gamma = 1/\sqrt K$. The global Grover operator $G_n$ defined in Eq. (\ref{def G n}) and the local Grover operator $G_m$ defined in Eq. (\ref{def G m}) can be reformulated as elements in the $O(3)$ group \cite{KV06}. Operators $G_{m}$ and $G_{n}$ have highly nontrivial commutation relations \cite{KV06}. The order of application of these operators is the key in the QPSA. Extensive studies have suggested that the optimal sequence (in oracle complexity) is $G_nG_m^{j_2}G_n^{j_1}$ \cite{KL06,KV06}. One can minimize the number of queries to the oracle (minimize $j_1+j_2+1$) given by a threshold success probability. The QPSA requires less number of oracles (the saved oracle number scales as $\sqrt b$) than Grover's algorithm. The QPSA can also be generalized into multitarget cases \cite{CK07,ZK18}. Interestingly, the QPSA can be performed in a hierarchical way: each time the QPSA finds several bits of the target bits $t$ \cite{KX07}.

\section{\label{sec:dep_opt}Depth Optimization}

\subsection{\label{subsec:MED}Minimal Expected Depth}

Depth is defined as the number of consecutive parallel gate operations. For example, the initial state $|s_n\rangle$ can be prepared with one depth circuit, see (\ref{def s n}). Suppose that the diffusion operator $D_n$ in Eq. (\ref{def I n}) has depth $\text{d}(D_n)$, which is the same as the depth of the $n$-qubit generalized Toffoli gate $\Lambda_{n-1}(X)$ \cite{NC10}. Different search tasks have different oracle realizations. We denote the ratio of  oracle depth $U_t$ and diffusion operator depth $D_n$ as $\alpha$:
\eq
\label{def alpha}
\alpha = \frac{\text{d}(U_t)}{\text{d}(D_n)}
\en
It is an important parameter for depth optimization. For the one-item search algorithm, the practical minimal value for $\alpha$ is 1: $\alpha\geq1$ \cite{FMLLDM17}. The ratio $\alpha$ maybe different for the same problem with a different database size. We fix $n$; then the ratio $\alpha$ is a constant for one problem. The design for a low-depth generalized Toffoli gate can also be a benefit for oracle depth \cite{GBDBRC19}. 

Given by $\text{d}(D_n)$ and $\alpha$, Grover's algorithm can be mapped to depth complexity directly. We define the \textit{minimal expected depth} (MED) of Grover's algorithm as:
\begin{equation}
  \label{def d G}
  \text d_\text{G}(\alpha) = \min_j \frac{\text d(G_n^j)}{P_n(j)}
\end{equation}
Here $P_n(j)$ defined in Eq. (\ref{def P n}) is the success probability of finding the target state (with $j$ Grover iterations). The numerator denotes the depth $\text d(G_n^j)=(\alpha+1)j\text d(D_n)$. The above optimization is the same as the expected iteration number optimization $j/P_n(j)$ \cite{BBHT98,GWC00}, up to a constant factor. Therefore, we can use $j_\text{exp}=\lfloor 0.583\sqrt N \rfloor$ in the MED. Note that we have $P_n(j_\text{exp})\approx 0.845$. Then we have
\begin{equation}
    \label{eqa d G}
    \text d_\text{G}(\alpha) \approx 0.69\times 2^{n/2}(\alpha+1)\text d(D_n)
\end{equation}
If the oracle can be constructed in polynomial depth $\text{d}(U_t)=\mathcal O(n^k)$, then the MED of Grover's algorithm scales as $\mathcal O(n^{k}2^{n/2})$ (assume that $k>1$). Grover's algorithm is optimal in oracle complexity \cite{BBHT98,Zalka99}. The minimal expected iteration number $j_\text{exp}$ is optimal. The scale $\mathcal O(n^{k}2^{n/2})$ is also optimal for depth complexity. However, we show that the number $\text d_\text{G}(\alpha)$ in Eq. (\ref{eqa d G}) is {\it not} optimal (if $\alpha$ in Eq. (\ref{def alpha}) is finite).

\subsection{\label{subsec: opt method} Optimization Method}

The local diffusion operator $D_m$ defined in Eq. (\ref{def I n m}) has lower depth than the global diffusion operator $D_n$ in Eq. (\ref{def I n}). The optimization idea is \textit{to replace the global diffusion operator by the local diffusion operator}. The global Grover operator $G_n$ defined in Eq. (\ref{def G n}) does not commute with the local Grover operator $G_m$ in Eq. (\ref{def G m}) \cite{KV06}. The order of $G_n$ and $G_m$ is important. Suppose that we have the sequence
\begin{equation}
\label{def S n m}
   S_{n,m}(j_1,j_2,\ldots,j_q) = G_n^{j_1}G_m^{j_{2}}\cdots G_n^{j_{q-1}}G_m^{j_q} 
\end{equation}
Here $\{j_1,j_2,\ldots,j_q\}$ are some non-negative integers. We have 
\begin{equation}
    j_\text{tot} = \sum_{p=1}^q j_{p}
\end{equation}
total number of queries to the oracle. To remove the ambiguity in the notation $S_{n,m}(j_1,j_2,\ldots,j_q)$, we require that \textit{the last number $j_q$ is always the number of local Grover operators}. For example, $S_{6,4}(1,2) = G_6G^2_4$ and $S_{6,4}(1,1,0) = G_4G_6$. Note that $S_{n,m}(j,0)=G^j_n$ is the original Grover algorithm. Since the sequence $S_{n,m}(j,0)=G^j_n$ does not have any local Grover operators, the number $m$ is irrelevant. As convention, we choose the notation $S_{n}(j,0)=S_{n,m}(j,0)$. The sequence $S_{n,m}(j_1,j_2,\ldots,j_q)$ can find the target state with probability:
\eq
P_{n,m}(j_1,j_2,\ldots,j_q) = |\langle t|S_{n,m}(j_1,j_2,\ldots,j_q)|s_n\rangle|^2
\en
Then we can define the expected depth of the $S_{n,m}(j_1,j_2,\ldots,j_q)$ algorithm. We want to minimize the expected depth, like for Grover's algorithm (\ref{def d G}). Define a new MED:
\begin{equation}
\label{d 1 G}
   {\rm d}_1(\alpha) = \min_{m,j_1,j_2,\ldots,j_q} \frac{{\rm d}(S_{n,m}(j_1,j_2,\ldots,j_q))}{P_{n,m}(j_1,j_2,\ldots,j_q)}
\end{equation}
The minimization goes through non-negative integers $\{j_1,j_2,\ldots,j_q\}$. We also optimize the number $m$ (positive integer), which is $m<n$. The minimal value for $m$ is 2. The subscript 1 defined in $\text d_1(\alpha)$ suggests that we find the target state in one stage, i.e., no measurement within the algorithm until the end. In the quantum circuit model, a one-stage algorithm means only three steps: initialization, unitary operations and measurements. We can define multistage algorithms, which have several rounds of initializations, unitary operations, and measurements. Later we define the MED of multistage search algorithms.

Let us see one example. For $n=6$, Grover's algorithm has the MED when $j=4$:
\begin{equation}
  P_6(4) = |\langle t|G_6^4|s_6\rangle|^2\approx 0.816
\end{equation}
Consider a new sequence:
\begin{equation}
    S_{6,4}(1,1,2) = G_4G_6G_4^2
\end{equation}
and $S_{6,4}(1,1,2)$ gives the success probability
\begin{equation}
    P_{6,4}(1,1,2) = |\langle t|S_{6,4}(1,1,2)|s_6\rangle|^2\approx 0.755
\end{equation}
Note that both sequences $G_6^4$ and $G_4G_6G_4^2$ have four oracles. According to \cite{BBCDMSSSW95}, six-qubit and four-qubit Toffoli gates can be decomposed into 64 and 16 depth circuits (with single- and two-qubit gates). We suppose that $\text d(D_6)=64$ and $\text d(D_4)=16$. One can find that if the ratio $\alpha$ in Eq. (\ref{def alpha}) is $\alpha<2.029$, then the new sequence $G_4G_6G_4^2$ has a lower expected depth. More examples (about the $n=6$ search algorithm) with quantum circuit diagrams can be found in Appendix \ref{Appendix examples 6}.

We can go back to Grover's algorithm if the number of $G_m$ is zero. We always have 
\begin{equation}
    \text d_1(\alpha)\leq\text d_\text{G}(\alpha)
\end{equation}
The choice of subspace (acted upon by local diffusion operators $D_m$ defined in (\ref{def I n m})) can be arbitrary, such as qubits with high connectivity in real quantum computers. But all local diffusion operators $D_m$ should act on the same qubits. For example, the sequence $S_{6,4}(1,1,2)$ has three local Grover operators. The three local diffusion operators are acting on the same four qubits. Making the wrong choice of the subspace can dramatically increase the number of invariant amplitude subspaces. Such a strategy may have some advantages in search algorithms, but it is beyond the scope of this paper. 
  
The minimization results will depend on: the size of the database (the number $n$), the ratio between oracle depth $\text d(U_t)$ and diffusion operator depth $\text d(D_n)$ (the value of $\alpha$ defined in (\ref{def alpha})); how $\text d(D_n)$ scales with $n$ (logarithmic, linear, or quadratic with $n$). In numerical optimizations, we can set some constraints which rule out the possibility $\text d_1(\alpha)<\text d_\text{G}(\alpha)$. For example, we can set the total number of $G_n$ to less than $\lfloor0.69\sqrt N\rfloor$; if the number of $G_n$ is $j$, then the number of $G_m$ should be less than $\lfloor(0.69\sqrt N-j)(\alpha+1)/\alpha\rfloor$. As examples, we find the optimal sequence for $n=4,5,\ldots,10$ with $\alpha=1$ (assuming $\mathcal O(n)$ depth of the $\Lambda_{n-1}(X)$ gate \cite{BBCDMSSSW95}). The estimated depths are plotted in Fig. \ref{fig 2}. Details of the corresponding optimal sequences and success probabilities can be found  in Appendix \ref{Appendix opt examples}. 

\begin{figure}
    \begin{center}
    	\includegraphics[width=\columnwidth]{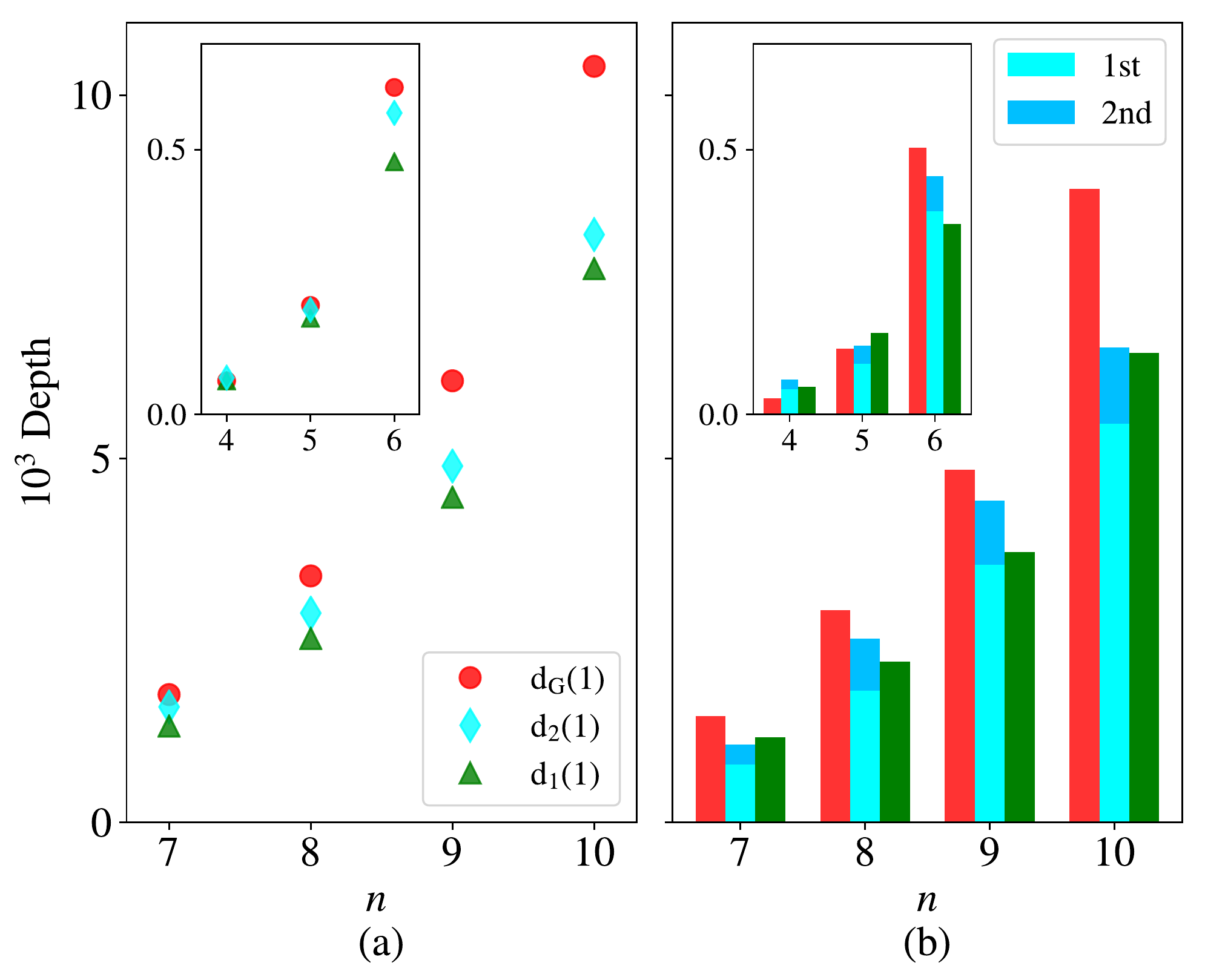}
    	\caption{(a) Estimated $\text d_\text{G}(\alpha)$ (MED of Grover's algorithm is defined in Eq. (\ref{def d G})), $\text d_1(\alpha)$ defined in Eq. (\ref{d 1 G}) and $\text d_2(\alpha)$ defined in Eq. (\ref{d 2 G}) with $\alpha=1$. Depth $\text d(D_n)$ is counted using the optimal results in Ref. \cite{BBCDMSSSW95}. The corresponding optimal sequences and success probabilities are listed in Appendix \ref{Appendix opt examples}. (b) Depth of the optimal sequence. The left (red) bar is Grover's algorithm. The right (green) bar is the optimal sequence from $\text d_1(1)$ defined in Eq. (\ref{d 1 G}). $\text d_2(1)$ has two stages: the bottom of the middle bar is the depth of the first stage circuit and the top of the middle bar is the depth of the second stage circuit. }\label{fig 2}
    \end{center}
\end{figure}

\subsection{\label{subsec:divded and conquer} Depth Optimizations for Multistage Quantum Search Algorithms}

In the NISQ era, errors can be suppressed if a long algorithm is divided into shorter pieces (by new initializations and measurements). Inspired by the hierarchy QPSA \cite{KX07}, we propose depth optimizations for the multistage quantum search algorithm. For simplicity, we consider the two-stage quantum search algorithm firstly.

Suppose that the target state is divided into two-parts:
\begin{equation}
\label{eqa t t1}
    |t\rangle=|t_1\rangle\otimes|t_2\rangle
\end{equation}
Suppose that the bit length of $t_1$ is $m_1$ and the bit length of $t_2$ is $m_2$. Note that we have $m_1+m_2=n$. After first stage, the search algorithm can find $|t_1\rangle$ with high probability. Based on the result on the first stage, we can rescale the database. After the second stage, the algorithm can find $|t_2\rangle$ with high probability (if $|t_1\rangle$ is found in the first stage). The algorithm has the following steps: 

\begin{enumerate}[Step 1:]
    
    \item Initialize the state to $|s_n\rangle$ defined in Eq. (\ref{def s n}).
    
    \item Perform the sequence  
        \begin{equation}
            S^{(1)}_{n,m_2}(j_1,j_2,\ldots,j_q) = G_n^{j_1}G_{m_2}^{j_{2}}\cdots G_n^{j_{q-1}}G_{m_2}^{j_q}
        \end{equation}
    on the initial state $|s_n\rangle$. The local diffusion operator $D_{m_2}$ (defined in $G_{m_2}$) is acting on $m_2$ qubits. 
    
    \item Measure the qubits (computational basis measurements) which do \textit{not} have the local diffusion operator $D_{m_2}$ acting on them. Suppose that we get the classical results: $t'_1\in\{0,1\}^{m_1}$. The probability that $t'_1=t_1$ is denoted as $P^{(1)}_{n,m_2}(j_1,j_2,\ldots,j_q)$.
    
    \item Initialize the state to
        $$ |t_1'\rangle\otimes|s_{m_2}\rangle
        $$
        Here $|s_{m_2}\rangle$ is the rescaled initial state:
        \begin{equation}
            |s_{m_2}\rangle = H^{\otimes m_2}|0\rangle^{\otimes m_2}
        \end{equation}
        
    \item Perform the sequence  
        \begin{equation}
            S^{(2)}_{m_2,m'}(j'_1,j'_2,\ldots,j'_q) = G_{m_2}^{j'_1}G_{m'}^{j'_2}\cdots G_{m_2}^{j'_{q-1}}G_{m'}^{j'_q}
        \end{equation}
    on the new initial state. We have $m'<m_2$. The diffusion operator $D_{m_2}$ (defined in $G_{m_2}$) is acting on $|s_{m_2}\rangle$. And the diffusion operator $D_{m'}$ is acting on the subspace of $|s_{m_2}\rangle$.
    
    \item Measure the qubits (computational basis measurements) which have the initial state $|s_{m_2}\rangle$. Suppose that we get the classical results: $t'_2\in\{0,1\}^{m_2}$. The probability that $t'_2=t_2$ is denoted as $P^{(2)}_{m_2,m'}(j'_1,j'_2,\ldots,j'_q)$.
    
    \item Verify the solution $|t'\rangle = |t_1'\rangle\otimes|t_2'\rangle$ by classical oracle. If the solution is the target item, then stop; if not, back to step 1.
\end{enumerate}

Steps 1-3 are the first stage: we find $t_1$ with high probability. Steps 4-6 are the second stage: we find the remaining bits of the target state. Step 7 is used to verify. Different sequences $S^{(1)}_{n,m_2}(j_1,j_2,\ldots,j_q)$ and $S^{(2)}_{m_2,m'}(j'_1,j'_2,\ldots,j'_q)$ give different success probabilities $P^{(1)}_{n,m_2}(j_1,j_2,\ldots,j_q)$ and $P^{(2)}_{m_2,m'}(j'_1,j'_2,\ldots,j'_q)$. We want to find the MED. The MED of the two-stages search algorithm is 
\begin{widetext}
\begin{equation}
\label{d 2 G}
   {\rm d}_2(\alpha) = \min_{m_2,m',j_1,\ldots,j_q,j'_1,\ldots,j'_q} \frac{{\rm d}(S^{(1)}_{n,m_2}(j_1,j_2,\ldots,j_q))+{\rm d}(S^{(2)}_{m_2,m'}(j'_1,j'_2,\ldots,j'_q))}{P^{(1)}_{n,m_2}(j_1,j_2,\ldots,j_q)P^{(2)}_{m_2,m'}(j'_1,j'_2,\ldots,j'_q)}
\end{equation}
\end{widetext}
We optimize the total expected depth. We do not optimize the expected stage depth, because we cannot verify the partial bit by neither classical nor quantum oracle. Note that $m_2$ is the bit length of $t_2$. We can either fix $m_2$ or optimize different choices of $m_2$. In the definition of ${\rm d}_2(\alpha)$, we optimize the choices of $m_2$. The second-stage algorithm is a rescaled version of the full search algorithm. Such a two-stage quantum search algorithm (with depth optimization) can be easily generalized to the multi-stage quantum search algorithm. 

As an example, let us consider the $n=4$ two-stage search algorithm. Grover's algorithm (one-stage search algorithm) has the success probability 
\begin{equation}
    P_4(3)=|\langle t|G_4^3|s_4\rangle|^2\approx 0.961
\end{equation}
In a two-stage search algorithm, we divide the target state into two parts: $|t\rangle=|t_1\rangle|t_2\rangle$. We choose the first-stage sequence as $S^{(1)}_{4,2}(1,1)=G_4 G_2$. Then we measure the two qubits which do \textit{not} have $D_2$ (defined in $G_2$) acting on them. The probability that the measurement results reveal $|t_1\rangle$ is
\begin{equation}
    P^{(1)}_{4,2}(1,1)\approx 0.953
\end{equation}
Suppose that the measurement results are $|t'_1\rangle$ after the first stage. Then we rescale the initial state as $|t'_1\rangle\otimes|s_2\rangle$. We choose the second stage sequence as $S^{(2)}_{2}(1,0)=G_2$. Recall that the two-qubit Grover's algorithm can find the target state with 100\% probability with one Grover operator. Therefore, the second-stage success probability is
\begin{equation}
    P^{(2)}_{2}(1,0)=1
\end{equation}
Then the total success probability is
\begin{equation}
    P^{(1)}_{4,2}(1,1)P^{(2)}_{2}(1,0)\approx 0.953
\end{equation}
The result is quite close to Grover's algorithm with the same number of oracles, but the depth in each stage is less than in Grover's algorithm. 

Another interesting example (two-stage $n=4$ search algorithm) is that the sequence $S^{(1)}_{4,2}(1,2)$ gives probability $1$ for finding $t_1$. Combined with the second-stage sequence $S^{(2)}_{2}(1,0)$, we find a new approach for the $n=4$ exact search algorithm \cite{Diao10}. We estimate $\text d_2(\alpha)$ with $\alpha=1$ for the $n=4,5\ldots,10$ search algorithms, see Fig. \ref{fig 2}. The corresponding optimal sequences are listed in Appendix \ref{Appendix opt examples}. See Appendix \ref{Appendix examples 6} for more examples (with quantum circuit diagrams) on two-stage quantum search algorithms. 

\section{\label{sec:alpha}Critical Ratios}

\subsection{\label{sec: one stage alpha}The Critical Ratio for the One-stage Algorithm}

Grover's algorithm is optimal in the number of queries to the oracle \cite{BBHT98,Zalka99}. Grover's algorithm is a one-stage search algorithm: no measurement occurs within the algorithm until the end. When $\alpha\rightarrow\infty$, we expect ${\rm d}_1(\alpha)=\text d_\text{G}(\alpha)$ (no local diffusion operators). Here ${\rm d}_1(\alpha)$ is defined in Eq. (\ref{d 1 G}). And $\text d_G(\alpha)$ defined in Eq. (\ref{def d G}) is the MED of Grover's algorithm. We define the critical alpha $\alpha_{c,1}$ for the one-stage search algorithm:
\eq
\label{def alpha c}
\alpha_{c,1}=\max\{\alpha|{\text d}_1(\alpha) <{\text d}_{\text G}(\alpha)\}
\en
The subscript 1 in $\alpha_{c,1}$ denotes the one-stage search algorithm. Below $\alpha_{c,1}$, the depth of Grover's algorithm is {\it not} optimal. Based on the depth optimization method proposed in Sec. \ref{subsec: opt method}, we can give an estimation of $\alpha_{c,1}$:
\begin{theo}
    \label{theorem 1} $\alpha_{c,1}=\mathcal O(n^{-1}2^{n/2})$. 
\end{theo}	
\begin{proof}
The MED $\text d_1(\alpha)$ defined in Eq. (\ref{d 1 G}) is a search algorithm with two different diffusion operators. One is the local diffusion operator $D_m$, see (\ref{def I n m}). The other is the global diffusion operator $D_n$, see (\ref{def I n}). The local diffusion operator $D_m$ is only acting on the subspace of the database. We can follow a three-dimensional subspace: the target state $|t\rangle$ defined in Eq. (\ref{basis t}); the normalized sum of nontarget states in the target block $|ntt\rangle$ defined in Eq. (\ref{basis ntt}); the normalized sum of rest states in the database $|u\rangle$ defined in Eq. (\ref{basis u}). The notations are taken from the QPSA, see Sec. \ref{subsec:QPSA} and \cite{Korepin05,KG06}.

Operators $G_n$ and $G_m$ only change the relative amplitudes of states $|t\rangle$, $|ntt\rangle$, and $|u\rangle$. Therefore, operators $G_n$ and $G_m$ are elements of the $O(3)$ group \cite{KV06}. It is interesting to see that operator $G_m$ can be viewed as a rescaled version of $G_n$. In the new basis $\{|t\rangle,|ntt\rangle,|u\rangle\}$, the sequence $S_{n,m}(j)=G^j_m$ (which only has local Grover operators $G_m$) has the representation
\begin{equation}
    S_{n,m}(j)=G^j_m=\left(\begin{array}{ccc}
	  \cos(2j\theta_2)	& \sin(2j\theta_2) &  0\\ 
	-\sin(2j\theta_2)	& \cos(2j\theta_2) &  0\\ 
	0	& 0 & 1
	\end{array} 
	\right)
\end{equation}
For example, the matrix element $\sin(2j\theta_2)$ is obtained from
\begin{equation}
    \sin(2j\theta_2) = \langle t|S_{n,m}(j)|ntt\rangle
\end{equation}
The angle is defined as 
\begin{equation}
    \sin\theta_2 = 1/\sqrt b,\quad b=2^m
\end{equation}

We want to estimate the critical ratio $\alpha_{c,1}$. We consider the sequence:
\begin{equation}
    S_{n,n-1}(1,1,1) = G_{n-1}G_{n}G_{n-1}
\end{equation}
Here we choose $m=n-1$. It means that the database is divided into two blocks. At the basis $\{|t\rangle,|ntt\rangle,|u\rangle\}$ defined in Eqs. (\ref{basis t})-(\ref{basis u}), the sequence $S_{n,n-1}(1,1,1)$ has the matrix representation
\begin{widetext}
    \begin{equation}
    \label{def S 111}
    S_{n,n-1}(1,1,1)=\left(
    \begin{array}{ccc}
    c^2(c^2-3s^2)	& cs(3c^2-s^2)(c^2-3s^2) & s(3c^2-s^2) \\ 
    -cs(3c^2-s^2)(c^2-3s^2)	& s^2(s^2-3c^2)  & c(c^2-3s^2) \\ 
    -s(3c^2-s^2)	& c(c^2-3s^2) & 0
    \end{array} 
    \right)
    \end{equation}
\end{widetext}
with short notations $c=\cos\theta_2$ and $s=\sin\theta_2$. Note that $\sin\theta_2=\sqrt{2/N}$ since we choose $m=n-1$. The matrix $S_{n,n-1}(1,1,1)$ has the eigenvalues:
\begin{equation}
    \lambda_0 = -1,\quad\quad \lambda_\pm = e^{\pm i\gamma}
\end{equation}
with
\begin{equation}
    \tan \gamma = \frac{\Delta}{1+\cos\theta_2},\quad
    \Delta = \sqrt {3-2\cos(6\theta_2)-\cos^2(6\theta_2)}
\end{equation}
The corresponding normalized eigenvectors are denoted as $|v_0\rangle$ (with eigenvalue $\lambda_0$) and $|v_\pm\rangle$ (with eigenvalue $\lambda_\pm$). States $|v_0\rangle$ and $|v_\pm\rangle$ have the form:
\begin{subequations}
    	\begin{align}
    	\label{eigen lambda 0} |v_0\rangle =& \frac 1 {\mathcal N_0}\left(0,1,\cos\theta_2(1-4\cos^2\theta_2)\right)^T,\\
    	\label{eigen lambda +-} |v_\pm\rangle =& \frac 1 {\mathcal N_\pm}\left(\mp i \sqrt{\frac{3+\cos 6\theta_2}{2}},\cos 3\theta_2,1\right)^T
    	\end{align}
\end{subequations}
The notation $T$ means transpose. $\mathcal N_0$ and $\mathcal N_\pm$ are normalizations. Note that the eigenvector $|v_0\rangle$ (with eigenvalue $-1$) is orthogonal to the target state, i.e., $\langle t|v_0\rangle=0$. We can view the operator $S_{n,n-1}(1,1,1)$ as rotation combined with reflection. Rotation is around an axis perpendicular to $|t\rangle$. The rotation angle is $\gamma$. Reflection is around a plane perpendicular to $|t\rangle$. Iteration $S_{n,n-1}(1,1,1)$ on the initial state gives
\begin{multline}
    \label{tilde G s n}
     \langle t|S^{\tilde j}_{n,n-1}(1,1,1)|s_n\rangle= \\ \lambda_+^{\tilde j}\langle t|v_+\rangle\langle v_+|s_n\rangle+\lambda_-^{\tilde j}\langle t|v_-\rangle\langle v_-|s_n\rangle
\end{multline}
We have $\langle t|v_\pm\rangle = \mp i /\sqrt 2$. Because $N=2^n$ is a large number, the angle $\theta_2$ is a small number. We can expand:
\begin{subequations}
\begin{align}
    &\gamma= 3\sqrt 2 \theta_2+\mathcal O\left(\theta^2_2\right),\\
    &\langle v_\pm|s_n\rangle = \frac 1 {\sqrt 2} +\mathcal O\left(\theta_2\right)
\end{align}
\end{subequations}
We substitute the above relations into Eq. (\ref{tilde G s n}). After some algebra, we can get the success probability of finding the target state:
\begin{equation}
    |\langle t|S^{\tilde j}_{n,n-1}(1,1,1)|s_n\rangle|^2 = \sin^2\left(3\sqrt 2\tilde j\theta_2\right)+\mathcal O(\theta_2)
\end{equation}
Because the sandwich sequence $S_{n,n-1}(1,1,1)$ has three oracles, we set $\tilde j=3j$. Then the probability difference between $S^{\tilde j}_{n,n-1}(1,1,1)$ and Grover's algorithm (with the same number of oracles) is
\begin{equation}
|\langle t|G^j_n|s_n\rangle|^2 -|\langle t|S^{\tilde j}_{n,n-1}(1,1,1)|s_n\rangle|^2=\delta>0
\end{equation}
Here $\delta$ is a small number:
\begin{equation}
    \delta = \mathcal O(2^{-n/2})
\end{equation}
Grover's algorithm (with $j$ Grover iterations) has success probability $P_n(j)$, see Eq. (\ref{def P n}).
Then the success probability for the $S^{\tilde j}_{n,n-1}(1,1,1)$ sequence (with $\tilde j=j/3$ iterations) is $P_n(j)-\delta$. If we want the new sequence $S^{\tilde j}_{n,n-1}(1,1,1)$ to have lower expected depth than Grover's algorithm, we can set
\begin{equation}
    \frac{3(\alpha+1)\text d(D_n)}{P_n(j)}>\frac{(3\alpha+1)\text d(D_n)+2\text d(D_{n-1})}{P_n(j)-\delta}
\end{equation}    
The left-hand side (times $j/3$) is the expected depth of Grover's algorithm. The right-hand side (times $\tilde j=j/3$) is the expect depth of the $S^{\tilde j}_{n,n-1}(1,1,1)$ algorithm. The above inequality gives
\begin{equation}
\alpha < \frac{2(\text d(D_n)-\text d(D_{n-1}))P_n(j)}{3\text d(D_n)\delta}
\end{equation}
The diffusion operator $D_n$ has the depth $\text d(D_n) = \mathcal O(n)$ or $\text d(D_n) = \mathcal O(n^2)$ \cite{BBCDMSSSW95}. Then we have
\begin{equation}
    \alpha_c = \mathcal O(n^{-1}2^{n/2})
\end{equation}
This is the end of the proof.
\end{proof}

As examples, we numerically estimate $\alpha_{c,1}$ defined in Eq. (\ref{def alpha c}) for $n=4,5,\ldots,10$ based on the linear depth of $D_n$, see Appendix \ref{Appendix alpha} and Table \ref{Table 4}. Below the critical ratio $\alpha_{c,1}$, at least two-third of the global diffusion operators $D_n$ can be replaced by $D_{n-1}$ (to have lower expected depth). The saved depth scales as $\mathcal O(2^{n/2})$.

\subsection{\label{sec: two stage alpha}The Critical Ratio for the Two-stage Algorithm}

Similar to the one-stage search algorithm, we can define the critical ratio for the two-stage algorithm:
\eq
\label{def alpha c 2}
\alpha_{c,2}=\max\{\alpha|{\text d}_2(\alpha) <{\text d}_{\text G}(\alpha)\}
\en
Here $\text d_2(\alpha)$ is the MED of the two-stage search algorithm, defined in Eq. (\ref{d 2 G}). The two-stage search algorithm has two measurements. After the first measurement, we reinitialize the state in the rescaled database. The amplified amplitude of the target state $|t\rangle$ is lost in the new initialization. One can argue that  
\begin{equation}
    {\rm d}_2(\alpha)>{\rm d}_{1}(\alpha),
\end{equation}
and it implies that  $\alpha_{c,2}<\alpha_{c,1}$.  Analytically, we can prove the following theorem.
\begin{theo} \label{theorem 2}
    $\lim_{N\rightarrow\infty}\alpha_{c,2} = 1+\sqrt 3 \approx 2.732$.
\end{theo}
\begin{proof}
Similar to the proof of Theorem \ref{theorem 1}, we construct a special sequence. Then we compare the expected depth of such a sequence with the expected depth of Grover's algorithm. Since we consider the two-stage search algorithm, we need two sequences for two stages. First, we assume that the target state $|t\rangle$ has two parts $|t\rangle=|t_1\rangle\otimes|t_2\rangle$, the same as in Eq. (\ref{eqa t t1}). And the bit length of $t_2$ is 2. For the first stage, we consider the sequence:
\begin{equation}
    S^{\tilde j}_{n,2}(1,1) = \left(G_nG_2\right)^{\tilde j}
\end{equation}
In the first stage (by the sequence $S^{\tilde j}_{n,2}(1,1)$), we find $t_1$ with high probability. The probability is denoted as $P^{(1)}_{n,2}$. In the second stage, we have a rescaled two-qubit search algorithm. One Grover operator $G_2$ can find the target state with $100\%$ probability. Therefore, the second stage has the sequence:
\begin{equation}
    S_{2}(1,0) = G_2
\end{equation}
The probability of finding $t_2$ is $P^{(2)}_{2}=1$.

In the basis $\{|t\rangle,|ntt\rangle,|u\rangle\}$ defined in Eqs. (\ref{basis t})-(\ref{basis u}), the sequence $S_{n,2}(1,1)$ has the matrix representation
\begin{equation}
S_{n,2}(1,1) = \frac 1 2\left(
\begin{array}{ccc}
 \cos 2\gamma	& \sqrt 3 &  \sin 2\gamma \\ 
\sqrt 3\cos 2\gamma	& -1 & \sqrt 3 \sin 2\gamma \\ 
-2\sin 2\gamma	& 0 & 2\cos 2\gamma
\end{array} 
\right)
\end{equation}
with $\sin \gamma = 2/\sqrt N$. We can easily find eigenvalues and eigenvectors of $S_{n,2}(1,1)$. Then we can have a matrix expression for $S^{\tilde j}_{n,2}(1,1)$. Applying $S^{\tilde j}_{n,2}(1,1)$ on the initial state $|s_n\rangle$ (Eq. \ref{def s n rewrite}), \begin{equation}
	|\langle u|S^{\tilde j}_{n,2}(1,1)|s_n\rangle|^2 = \cos^2(\sqrt 3 \tilde j\gamma)+\mathcal O(\gamma)
\end{equation}
Note that $|\langle u|S^{\tilde j}_{n,2}(1,1)|s_n\rangle|^2$ is the probability of finding the state in the nontarget block. In other words, we have
\begin{equation}
    P^{(1)}_{n,2} = 1-|\langle u|S^{\tilde j}_{n,2}(1,1)|s_n\rangle|^2
\end{equation}
The second stage has probability 1 (the two-qubit Grover's algorithm with one Grover operator has probability 1). Then $P^{(1)}_{n,2}$ is also the probability of finding the target state. 

The two stages designed above have a total of $2\tilde j+1$ queries to the oracle. In order to compare with Grover's algorithm, we set $j=\sqrt 3 \tilde j$ (where $j$ is the number of queries to the oracle in Grover's algorithm). Grover's algorithm with $j$ iterations has a success probability $P_n(j)$ of finding the target state, see Eq. (\ref{def P n}). Then the two-stage search algorithm (with sequences $S^{\tilde j}_{n,2}(1,1)$ and $S_{2}(1,0)$) can find the target state with probability $P_n(j)+\delta$. Here $\delta$ is a small number in order $\delta = \mathcal O(2^{-n/2})$. If we want the two-stage search algorithm to have lower expected depth than Grover's algorithm, we need 
\begin{equation}
    \frac{(\alpha+1)\text d(D_n)}{P_n(j)}>\frac{(2\alpha+1)\text d(D_n)+3}{\sqrt 3(P_n(j)+\delta)}
\end{equation}
The left-hand side (times $j$) is the expected depth of Grover's algorithm (with $j$ iterations). The right-hand side (times $j$) gives the expected depth of the designed two-stage search algorithm. Note that the second-stage circuit only contributes order $\mathcal O(2^{-n/2})$ to the critical value $\alpha_{c,2}$; therefore, we can neglect it here. Then we can solve the inequality
\begin{equation}
    \alpha > 1+\sqrt 3 -\frac 3 {\text d(D_n)}+\mathcal O\left(2^{-n/2}\right)
\end{equation}
For large $N$, we have the critical ratio
\begin{equation}
    \lim_{N\rightarrow\infty}\alpha_{c,2} = 1+\sqrt 3 \approx 2.732
\end{equation}
This ends of the proof.
\end{proof}

Theorem $\ref{theorem 2}$ suggests that the two-stage search algorithm can have lower expected depth than Grover's algorithm, only when the oracle can be realized as efficiently as the global diffusion operator. The real advantage of the two-stage algorithm is to mitigate the error accumulations for long circuits. For examples, see Fig. \ref{fig 2} and Appendixes \ref{Appendix examples 6} and \ref{Appendix opt examples}. We numerically estimate the value $\alpha_{c,2}$ ($n=4,5,\ldots,10$) based on a linear scale depth of $\text d(D_n)$, see Appendix \ref{Appendix alpha} and Table \ref{Table 4}.

\section{\label{sec:parallel}Parallel Running of Quantum Search Algorithm}

Now we discuss how to run the quantum search algorithm on several quantum computers in parallel. The simplest idea is running a low-success-probability search algorithm on different quantum computers. We verify the result by classical oracle and continue the algorithm until one of the quantum computers finds the target state \cite{GWC00}. First we can set a threshold success probability. Then we find the optimal sequence which gives the MED (the success probability is lower than the threshold success probability). We can run such a sequence on several quantum computers. 

Another parallel running method is to combine the random guess with search algorithm, as mentioned in Ref. \cite{Korepin05} for the QPSA. For example, the target state is divided into two parts: $|t\rangle=|t_1\rangle\otimes|t_2\rangle$, the same as in Eq. (\ref{eqa t t1}). One can randomly guess the bits $t_1$. Then one performs the search algorithm on bits $t_2$. Each quantum computer can pick up one guess. However, if more than half of the bits are chosen randomly, the quadratic speedup is lost. Such a strategy is more efficient if some of the bits have higher probability (prior information about the target state). 
    	
If we want near-deterministic (the fail probability is $\mathcal O(2^{-n/2})$) parallel running of the search algorithm, then we can apply the multistage search algorithm on different quantum computers. Suppose the target state has length $n$. The target state is divided into $p$ parts, and each part has equal $n/p$ length. Then we can assign the search algorithm on $p$ quantum computers. Each quantum computer finds one part of the target state. Combining all the results from each quantum computers, we can piece together the whole solution $t$ at one time. The sequence running on each quantum computer can be found by maximizing the number of local Grover operators $G_m$ defined in Eq. (\ref{def G m}), based on some threshold success probability ($\mathcal O(1-2^{-n/2})$). It requires at most $n$ quantum computers. Each quantum computer finds one bit of the target state. However, the most efficient way to find one bit of the target state is by running the random-guess one-bit search algorithm \cite{Korepin05}. 

\section{\label{sec:conclusion}Conclusion and Outlook}

In this paper, we propose a new way to optimize the depth of quantum search algorithms. The quantum search algorithm can be realized by global and local diffusion operators. The ratio of the depth of the oracle and  global diffusion operator is important. The ratio is denoted by $\alpha$, and  defined in Eq. (\ref{def alpha}). The minimal practical value for $\alpha$ is 1 (in one target search algorithm). When $\alpha$ is below a threshold, we can design a new algorithm (new sequence) which has a lower expected depth than Grover's algorithm. We gave examples for $\alpha=1$. In examples, our algorithm has around $20\%$ lower depth than Grover's algorithm. We also study the depth optimization in the multi-stage quantum search algorithm. In each stage, the circuit has lower depth than in Grover's algorithm. The multistage quantum search algorithm gives a natural way for parallel running of the quantum search algorithm. 

Ideas in this work can be easily generalized to the multitarget solution search \cite{BBHT98}. However, the exact number of target states is required in order to find the optimal sequence. In this paper, we only consider two kinds of diffusion operators (at each stage). Further improvement is possible if more diffusion operators are working together. It will be interesting to optimize the depth of the amplitude amplification algorithm \cite{Grover98,BHMT00}. Grover's algorithm is only optimal in the oracle measure. Our search algorithm has lower depth than Grover's algorithm.

\begin{acknowledgments}

The authors are grateful to Professor Jin Wang and Yulun Wang. V.K. is supported by SUNY Center for Quantum Information Science at Long Island Project No. CSP181035.
	
\end{acknowledgments}

\appendix

\section{\label{Appendix examples 6} Example for \texorpdfstring{$n=6$}{Lg} Search Algorithm with Depth Optimization}

Different problems have different oracles. For demonstration, we can consider the simplest oracle. As mentioned in Ref. \cite{FMLLDM17}, the oracle is single-qubit-gate equivalent to the $n$-qubit Toffoli gate $\Lambda_{n-1}(X)$. Suppose $|t\rangle = |000000\rangle$ ($n=6$). We can have the oracle:
\begin{equation*}
\Qcircuit @C=1em @R=0.982em {
& \multigate{5}{~U_t~} & \qw &\\
& \ghost{~U_t~} & \qw &\\
& \ghost{~U_t~} & \qw &\\
& \ghost{~U_t~} & \qw & \raisebox{0.7cm}{\hspace{0.3cm}=}\\
& \ghost{~U_t~} & \qw &\\
& \ghost{~U_t~} & \qw &
}
\hspace{0.7cm}
\Qcircuit @C=0.4em @R=0.6em {
& \gate{X} & \qw & \ctrl{1} & \qw & \gate{X} & \qw\\
& \gate{X} & \qw & \ctrl{1} & \qw & \gate{X} & \qw\\
& \gate{X} & \qw & \ctrl{1} & \qw & \gate{X} & \qw\\
& \gate{X} & \qw & \ctrl{1} & \qw & \gate{X} & \qw\\
& \gate{X} & \qw & \ctrl{1} & \qw & \gate{X} & \qw\\
& \gate{X} & \gate{H} & \targ & \gate{H} & \gate{X} & \qw
}
\vspace{0.3cm}
\end{equation*}
According to Ref. \cite{BBCDMSSSW95}, the $\Lambda_5(X)$ gate can be realized by a depth of 61 circuit: $\text d(\Lambda_5(X))=61$ (if the quantum computer can perform any single-qubit gates and any two-qubit controlled gates). In real quantum computers, the depth $\text d(\Lambda_5(X))$ may be much larger since not all qubits are connected. Nevertheless, we can set
\begin{equation}
    \text d(U_t) = \text d(\Lambda_5(X))+2 = 63
\end{equation}

The global diffusion operator ($n=6$) is also single-qubit-gate equivalent to the six-qubit Toffoli gate $\Lambda_5(X)$. We have
\begin{equation*}
\Qcircuit @C=1em @R=0.982em {
& \multigate{5}{~D_6~} & \qw &\\
& \ghost{~D_6~} & \qw &\\
& \ghost{~D_6~} & \qw &\\
& \ghost{~D_6~} & \qw & \raisebox{0.7cm}{\hspace{0.3cm}=}\\
& \ghost{~D_6~} & \qw &\\
& \ghost{~D_6~} & \qw &
}
\hspace{0.7cm}
\Qcircuit @C=0.4em @R=0.6em {
& \gate{H} & \gate{X} & \ctrl{1} & \gate{X} & \gate{H} & \qw\\
& \gate{H} & \gate{X} & \ctrl{1} & \gate{X} & \gate{H} & \qw\\
& \gate{H} & \gate{X} & \ctrl{1} & \gate{X} & \gate{H} & \qw\\
& \gate{H} & \gate{X} & \ctrl{1} & \gate{X} & \gate{H} & \qw\\
& \gate{H} & \gate{X} & \ctrl{1} & \gate{X} & \gate{H} & \qw\\
& \gate{Z} & \qw & \targ & \qw & \gate{Z} & \qw
}
\vspace{0.3cm}
\end{equation*}
Therefore, we can set
\begin{equation}
    \text d(D_6) = \text d(\Lambda_5(X))+2 = 63
\end{equation}
Therefore, we have the ratio $\alpha=1$, see Eq. (\ref{def alpha}). The local diffusion operators are acting on the subspace of six qubits. For example, the $D_4$ diffusion operator has the quantum circuit diagram
\begin{equation*}
\Qcircuit @C=1em @R=0.982em @!R {
& \multigate{3}{~D_4~} & \qw &\\
& \ghost{~D_4~} & \qw &\\
& \ghost{~D_4~} & \qw &\raisebox{0.7cm}{\hspace{0.3cm}=}\\
& \ghost{~D_4~} & \qw &
}
\hspace{0.7cm}
\Qcircuit @C=0.4em @R=0.6em @!R {
& \gate{H} & \gate{X} & \ctrl{1} & \gate{X} & \gate{H} & \qw\\
& \gate{H} & \gate{X} & \ctrl{1} & \gate{X} & \gate{H} & \qw\\
& \gate{H} & \gate{X} & \ctrl{1} & \gate{X} & \gate{H} & \qw\\
& \gate{Z} & \qw & \targ & \qw & \gate{Z} & \qw
}
\end{equation*}
And the local diffusion operator $D_2$ is single-qubit-gate equivalent to the CNOT gate:
\begin{equation*}
\Qcircuit @C=1em @R=0.982em @!R {
& \multigate{1}{~D_2~} & \qw & \\
& \ghost{~D_2~} & \qw & \raisebox{0.7cm}{\hspace{0.3cm}=}
}
\hspace{0.7cm}
\Qcircuit @C=0.4em @R=0.6em @!R {
& \gate{H} & \gate{X} & \ctrl{1} & \gate{X} & \gate{H} & \qw\\
& \gate{Z} & \qw & \targ & \qw & \gate{Z} & \qw
}
\end{equation*}
Accordingly, we have
\begin{align}
\text d(D_4) &= \text d(\Lambda_3(X))+2 = 15, \\
\text d(D_2) &= \text d(\Lambda_1(X))+2 = 3
\end{align}

Near-term quantum (or NISQ) computers are subjected to limited coherence time. We have to design a low depth algorithm, or divide a long circuit into shorter pieces. In the case of the $n=6$ search algorithm, Grover's algorithm needs six iterations to give the maximal probability of finding the target state. In experiments, we do not need to run the quantum search algorithm until the maximal probability is reached. For low depth consideration, we give examples of search algorithms with one or two oracles. Even in such simple scenarios, we can do better by using local diffusion operators. 

\subsection{One-oracle Algorithm}

\begin{itemize}
    \item Grover's algorithm. The one-iteration Grover's algorithm gives
    \begin{equation*}
    \Qcircuit @C=0.5em @R=0.6em @!R {
    \lstick{|0\rangle} & \gate{H} & \multigate{5}{U_t} & \multigate{5}{D_6} & \meter \\
    \lstick{|0\rangle} & \gate{H} & \ghost{U_t} & \ghost{D_6} & \meter \\ 
    \lstick{|0\rangle} & \gate{H} & \ghost{U_t} & \ghost{D_6} & \meter \\
    \lstick{|0\rangle} & \gate{H} & \ghost{U_t} & \ghost{D_6} & \meter \\
    \lstick{|0\rangle} & \gate{H} & \ghost{U_t} & \ghost{D_6} & \meter \\
    \lstick{|0\rangle} & \gate{H} & \ghost{U_t} & \ghost{D_6} & \meter 
    }
    \end{equation*}
    Measurements at the end are computational basis measurements. The whole circuit has depth 
    \begin{equation}
        \text d(G_6)=126
    \end{equation}
    We can incorporate the initial Hadamard gates into $G_6$. The success probability of finding the target state is
    \begin{equation}
        P_6(1) = |\langle t|G_6|s_6\rangle|^2 \approx 0.1348
    \end{equation}
    The result is better than that of the classical algorithm. The optimal classical search has a success probability of $3.15\%$: a single query followed by a random guess if the query fails ($1/64+1/63\approx 3.15\%$). To evaluate the efficiency, we can calculate the expected depth:
    \begin{equation}
        \frac{\text d(G_6)}{P_6(1)} \approx 935
    \end{equation}
    
    \item Our optimized algorithm. In order to lower the depth, we can apply, for example, one iteration of the local operator $G_4$. The one-iteration local Grover operator has the circuit
    \begin{equation*}
    \Qcircuit @C=0.5em @R=0.6em @!R {
    \lstick{|0\rangle} & \gate{H} & \multigate{5}{U_t} & \qw & \meter \\
    \lstick{|0\rangle} & \gate{H} & \ghost{U_t} & \qw & \meter \\ 
    \lstick{|0\rangle} & \gate{H} & \ghost{U_t} & \multigate{3}{D_4} & \meter \\
    \lstick{|0\rangle} & \gate{H} & \ghost{U_t} & \ghost{D_4} & \meter \\
    \lstick{|0\rangle} & \gate{H} & \ghost{U_t} & \ghost{D_4} & \meter \\
    \lstick{|0\rangle} & \gate{H} & \ghost{U_t} & \ghost{D_4} & \meter 
    }
    \end{equation*}
    Note that $S_{6,4}(1)=G_4$ is still a six-qubit gate, although $D_4$ is a four-qubit gate. For notation about $S_{6,4}(1)$, see Eq. (\ref{def S n m}). The whole circuit has depth
    \begin{equation}
        \text d(G_4)=78
    \end{equation}
    The depth is lower compared with that of $G_6$. The success probability of finding the target state is
    \begin{equation}
        P_{6,4}(1)=|\langle t|S_{6,4}(1)|s_6\rangle|^2 \approx 0.1181
    \end{equation}
    The success probability decreases a little bit, but still outperforms the classical case. The expected depth is:
    \begin{equation}
        \frac{\text d(S_{6,4}(1))}{P_{6,4}(1)} \approx 660
    \end{equation}
    \textbf{The circuit is $\bm{38\%}$ shorter than one $\bm{G_6}$ iteration. The expected depth is $\bm{29\%}$ lower.} The local diffusion operator may decrease the success probability, but it saves depth.
\end{itemize}

\begin{table*}[h]
	\begin{ruledtabular}
	\caption{\label{Table 1}  Estimated MED of Grover's algorithm, based on $\alpha=1$. The number $\alpha$ (defined in Eq. (\ref{def alpha})) is the ratio between oracle depth and diffusion operator depth. Diffusion operators $D_n$ have depth $\text d(D_n)=\{16,32,64,123,163,203,243\}$ with $n=4,5,\ldots,10$, which comes from the decomposition of an $n$-qubit Toffoli gate \cite{BBCDMSSSW95}. Single-run depth is the depth of the optimal sequence (without considering the success probability). The MED $\text d_\text{G}(\alpha=1)$ is defined in Eq. (\ref{def d G}). The notation $S_n(j,0)$ means $G^j_n$.}

		\begin{tabular}{ccccc}
			
			$n$ & Optimal sequence & Success probability &  Single-run depth & $\text d_\text{G}(1)$  \\ \hline
			
			4 & $S_4(1,0)$ & 0.473 & 30 & 63.47  \\
			
			5 & $S_5(2,0)$ & 0.602 & 124 & 205.83 \\
			
			6 & $S_6(4,0)$ & 0.816 & 504 & 617.36 \\
			
			7 & $S_7(6,0)$ & 0.833 & 1464 & 1756.35 \\
			
			8 & $S_8(9,0)$ & 0.861 & 2916 & 3388.03 \\
			
			9 & $S_9(12,0)$ & 0.798 & 4848 & 6071.76 \\
			
			10 & $S_{10}(18,0)$ & 0.838 & 8712 & 10397.28 \\
			
		\end{tabular}
    \end{ruledtabular} 
    
\end{table*}

\begin{table*}[h]

    \begin{ruledtabular}
	\caption{\label{Table 2} MED of one-stage search algorithm optimized by local diffusion operators, based on $\alpha=1$. The MED $\text d_1(\alpha=1)$ is defined in Eq. (\ref{d 1 G}). The depth of the diffusion operator is $\text d(D_n)=\{8,16,32,64,123,163,203,243\}$ with $n=3,4,\ldots,10$. The sequence notation means $S_{n,m}(j_1,j_2,\ldots,j_q) = G_n^{j_q}G_m^{j_{q-1}}\cdots G_n^{j_2}G_m^{j_1}$, see Eq. (\ref{def S n m}), and $j_q$ is always the number of the local diffusion operator.}
	    \begin{tabular}{ccccc}
		
		$n$ & Optimal sequence & Success probability &  Single-run depth & $\text d_1(1)$  \\ \hline
		
		4 & $S_{4,3}(1,1)$ & 0.821 & 52 & 63.32  \\
		
		5 & $S_{5,4}(1,1,1)$ & 0.849 & 154 & 181.48 \\
		
		6 & $S_{6,4}(1,1,2)$ & 0.755 & 360 & 476.97 \\
		
		7 & $S_{7,4}(1,1,2,1,2)$ & 0.887 & 1173 & 1322.75 \\
		
		8 & $S_{8,4}(1,1,2,1,2,1,2)$ & 0.875 & 2211 & 2527.43 \\
		
		9 & $S_{9,5}(1,1,2,1,2,1,2,1,2)$ & 0.831 & 3713 & 4470.20 \\
		
		10 & $S_{10,5}(1,1,2,1,2,1,2,1,2,1,2,1,2)$ & 0.847 & 6453 & 7614.56 \\
		
	    \end{tabular}
	\end{ruledtabular}
\end{table*}

\begin{table*}[h]
	\begin{ruledtabular}
	\caption{\label{Table 3} MED of the two-stage search algorithm, based on $\alpha=1$. The MED $\text d_2(\alpha=1)$ is defined in Eq. (\ref{d 2 G}). The depth of the diffusion operator is $\text d(D_n)=\{4,8,16,32,64,123,163,203,243\}$ with $n=2,3,4,\ldots,10$.}
	\begin{tabular}{cccccccc}
		
		$n$ & \multicolumn{2}{c}{Optimal sequence} & \multicolumn{2}{c}{Success probability} &  \multicolumn{2}{c}{Single-run depth} & $\text d_2(1)$  \\
		
		& Stage 1 & Stage 2 & Stage 1 & Stage 2 & Stage 1 & Stage 2 & \\ \hline
		
		4 & $S_{4,2}(1,1)$ & $S_{2}(1,0)$ & 0.953 & 1 & 48 & 18 & 69.25 \\
		
		5 & $S_{5,2}(1,1)$ & $S_{2}(1,0)$ & 0.658 & 1 & 96 & 34 & 197.51 \\
		
		6 & $S_{6,2}(1,1,1,1)$ & $S_{2}(1,0)$ & 0.791 & 1 & 384 & 66 & 569.22 \\
		
		7 & $S_{7,4}(1,4)$ & $S_{4}(2,0)$ & 0.739 & 0.908 & 792 & 274 & 1587.09 \\
		
		8 & $S_{8,5}(1,4,1,2)$ & $S_{5,4}(1,1,2)$ & 0.882 & 0.998 & 1806 & 724 & 2876.40 \\
		
		9 & $S_{9,5}(1,4,1,3,1,3)$ & $S_{5,4}(1,1,2)$ & 0.906 & 0.998 & 3542 & 884 & 4898.88 \\
		
		10 & $S_{10,5}(1,4,1,3,1,3,1,3)$ & $S_{5,4}(1,1,2)$ & 0.810 & 0.998 & 5485 & 1044 & 8081.89 \\

	    \end{tabular}
	\end{ruledtabular}
\end{table*} 

\begin{table*}[h]
	\begin{ruledtabular}
	\caption{\label{Table 4} Numerical values for critical ratios $\alpha_{c,1}$ in Eq. (\ref{def alpha c}) and $\alpha_{c,2}$ in Eq. (\ref{def alpha c 2}). The results are based on the linear scale depth of the diffusion operator $\text d(D_n)$, see Ref. \cite{BBCDMSSSW95}. Theorem \ref{theorem 1} shows that $\alpha_{c,1}$ scales as $\mathcal O(n^{-1}2^{n/2})$. Theorem \ref{theorem 2} shows that $\alpha_{c,2}$ approaches $1+\sqrt 3$ when $N=2^n$ is very large.}
	
	\begin{tabular}{cccccccc}
		
		$n$ & 4 & 5 & 6 & 7 & 8 & 9 & 10  \\ \hline
		
		$\alpha_{c,1}$ & 2.07 & 4.64 & 14.65 & 29.45 & 32.88 & 45.95 & 83.97  \\
		
		$\alpha_{c,2}$ & NA & 1.21 & 1.53 & 1.76 & 2.00 & 2.17 & 2.28 \\

	\end{tabular}
	
	\end{ruledtabular}
\end{table*} 

\subsection{Two-oracle Algorithm}

We can apply same strategy for the two-iteration search algorithm: design a circuit with local diffusion operators and find the optimal one with the least expected depth. We can also design a two-stage quantum search algorithm. And for each stage we use two oracles. 

\begin{itemize}
    \item Grover's algorithm. The two-iteration Grover's algorithm gives:
    \begin{equation*}
    \Qcircuit @C=0.5em @R=0.6em {
    \lstick{|0\rangle} & \gate{H} & \multigate{5}{G_6} & \multigate{5}{G_6} & \meter \\
    \lstick{|0\rangle} & \gate{H} & \ghost{G_6} & \ghost{G_6} & \meter \\ 
    \lstick{|0\rangle} & \gate{H} & \ghost{G_6} & \ghost{G_6} & \meter \\
    \lstick{|0\rangle} & \gate{H} & \ghost{G_6} & \ghost{G_6} & \meter \\
    \lstick{|0\rangle} & \gate{H} & \ghost{G_6} & \ghost{G_6} & \meter \\
    \lstick{|0\rangle} & \gate{H} & \ghost{G_6} & \ghost{G_6} & \meter 
    }
    \end{equation*}
    The whole circuit has depth
    \begin{equation}
        \text d(G^2_6)=252
    \end{equation}
    The success probability of finding the target state is
    \begin{equation}
        P_6(2,0)=|\langle t|G^2_6|s_6\rangle|^2 \approx 0.3439,
    \end{equation}
    and the expected depth is
    \begin{equation}
        \frac{\text d(G^2_6)}{P_6(2)}\approx 733
    \end{equation}
    
    \item Our two-stage search algorithm. We divide the target state into two parts: $|t_1\rangle$ and $|t_2\rangle$. Here $t_1$ is two bits long and $t_2$ is four bits long. Accordingly, we can design a search algorithm which has two stages: the first stage finds $|t_1\rangle$ and the second stage finds $|t_2\rangle$. In each stage, we only have two Grover operators (local or global Grover operators). 
    
    The first stage has the sequence $S^{(1)}_{6,4}(1,1,0)=G_4G_6$. We have the circuit diagram:
    \begin{equation*}
    \Qcircuit @C=0.5 em @R=.6em @!R {
    \lstick{|0\rangle} & \gate{H} & \multigate{5}{U_t} & \multigate{5}{D_6} & \multigate{5}{U_t} & \qw & \meter \\
    \lstick{|0\rangle} & \gate{H} & \ghost{U_t} & \ghost{D_6} & \ghost{U_t} & \qw & \meter \\ 
    \lstick{|0\rangle} & \gate{H} & \ghost{U_t} & \ghost{D_6} & \ghost{U_t} &  \multigate{3}{D_4} & \qw \\
    \lstick{|0\rangle} & \gate{H} & \ghost{U_t} & \ghost{D_6} & \ghost{U_t} & \ghost{D_4} & \qw \\
    \lstick{|0\rangle} & \gate{H} & \ghost{U_t} & \ghost{D_6} & \ghost{U_t} & \ghost{D_4} & \qw \\
    \lstick{|0\rangle} & \gate{H} & \ghost{U_t} & \ghost{D_6} & \ghost{U_t} & \ghost{D_4} & \qw  \gategroup{1}{3}{6}{4}{.4em}{--} \gategroup{1}{5}{6}{6}{.4em}{--}  \\
    & & \raisebox{.1cm}{\hspace{1cm} $G_6$ } & & \raisebox{.1cm}{\hspace{1cm} $G_4$ }
    } 
    \end{equation*}
    We only measure the qubit which does \textit{not} have $D_4$ (defined in $G_4$) performed. The probability of finding $|t_1\rangle$ is $P^{(1)}_{6,4}(1,1,0)$:
    \begin{equation}
        P^{(1)}_{6,4}(1,1,0) \approx 0.5604
    \end{equation}
    The first-stage circuit has depth 
    \begin{equation}
        \text d(S^{(1)}_{6,4}(1,1,0))=204
    \end{equation}
    In the first stage, suppose that the two classical measurement bits are $b_1$ and $b_2$ ($b_1,b_2\in \{0,1\}$). We cannot verify the partial bits $b_1$ and $b_2$. Since $P_{6,4}(1,1,0)>1/2$, the majority vote can be applied. 
    
    In the second stage, we choose the sequence
    \begin{equation}
        S^{(2)}_{4}(2,0) = G^2_4
    \end{equation}
    And we have the circuit
    \begin{equation*}
    \Qcircuit @C=0.5em @R=0.6em @!R {
    \lstick{|0\rangle} & \gate{X^{b_1}} & \multigate{5}{U_t} & \qw & \multigate{5}{U_t} & \qw  & \qw \\
    \lstick{|0\rangle} & \gate{X^{b_2}} & \ghost{U_t} & \qw & \ghost{U_t} & \qw  & \qw \\ 
    \lstick{|0\rangle} & \gate{H} & \ghost{U_t} & \multigate{3}{D_4} & \ghost{U_t} &  \multigate{3}{D_4} & \meter \\
    \lstick{|0\rangle} & \gate{H} & \ghost{U_t} & \ghost{D_4} & \ghost{U_t} & \ghost{D_4} & \meter \\
    \lstick{|0\rangle} & \gate{H} & \ghost{U_t} & \ghost{D_4} & \ghost{U_t} & \ghost{D_4} & \meter \\
    \lstick{|0\rangle} & \gate{H} & \ghost{U_t} & \ghost{D_4} & \ghost{U_t} & \ghost{D_4} & \meter \gategroup{1}{3}{6}{4}{.4em}{--} \gategroup{1}{5}{6}{6}{.4em}{--} \\
    & & \raisebox{.1cm}{\hspace{1cm} $G_4$ } & & \raisebox{.1cm}{\hspace{1cm} $G_4$ }
    } 
    \end{equation*}
    The initial state is the rescaled database. For example, in the first stage we find the $|01\rangle$ state; then we prepare the input $|01\rangle\otimes H^{\otimes 4}|0\rangle^{\otimes4}$. The probability of finding $|t_2\rangle$ is $P^{(2)}_{4}(2,0)$:
    \begin{equation}
        P^{(2)}_{4}(2,0)\approx 0.9084
    \end{equation}
    The second-stage circuit has depth 
    \begin{equation}
        \text d(S^{(2)}_{4}(2,0)) = 156
    \end{equation} 
    We have the expected depth
    \begin{equation}
        \frac{\text d(S^{(1)}_{6,4}(1,1,0))+\text d(S^{(2)}_{4}(2,0))}{P^{(1)}_{6,4}(1,1,0)P^{(2)}_{4}(2,0)}\approx 707
    \end{equation}
    \textbf{The expected depth is still $\bm{4.50\%}$ lower than that of the two-iteration Grover's algorithm. And the first stage has $\bm{19.05\%}$ shorter depth and the second stage has $\bm{38.10\%}$ shorter depth.} Besides, the two-stage strategy is subjected to half the errors from measurements. 
    
\end{itemize}

\section{\label{Appendix opt examples} Optimal Sequences Based on \texorpdfstring{$\alpha=1$}{Lg}}

We present detailed numerical results plotted in Fig. \ref{fig 2}. Suppose that we have quantum computers equipped with arbitrary single-qubit gates and arbitrary controlled two-qubit gates. It is well known that the $n$-qubit Toffoli gate $\Lambda_{n-1}(X)$ can be linearly decomposed into basic operators with one ancillary qubit \cite{BBCDMSSSW95}. We set the depth of the $n$-qubit Toffoli gate as $\text d(\Lambda_{n-1}(X))=\{1,5,13,29,61,120,160,200,240\}$ with $n=2,3,\ldots,10$, see Ref. \cite{BBCDMSSSW95}. Then the depth of the diffusion operator $D_n$ (\ref{def I n}) is
\begin{equation}
    \text d(D_n)=\text d(\Lambda_{n-1}(X))+2
\end{equation}
See Fig. \ref{fig G n m}. The depth of the oracle $U_t$ is characterized by the ratio $\alpha=\text d(U_t)/\text d(D_n)$. The ratio $\alpha$ is defined in Eq. (\ref{def alpha}). As an example, we set $\alpha=1$. The ratio $\alpha=1$ implies the simplest oracle construction, see Ref. \cite{FMLLDM17}. We list the optimal strategy (with the MED) of Grover's algorithm ($n=4,5,\ldots,10$) in Table \ref{Table 1}. When $N=2^n$ is large, the optimized iteration number in Grover's algorithm converges to $\lfloor 0.583\sqrt N\rfloor$, and the success probability converges to 0.844. The optimizations are independent of $\alpha$, see $\text d_\text{G}(\alpha)$ in Eq. (\ref{eqa d G}).

We numerically find the optimal sequence (optimized by the local diffusion operator). Similarly, we set $\alpha=1$. The MED is given by $\text d_1(\alpha=1)$, see Eq. (\ref{d 1 G}). The results are listed in Table \ref{Table 2}. We also numerically find the optimal sequence for the two-stage search algorithm. The MED is given by $\text d_2(\alpha=1)$, see Eq. (\ref{d 2 G}). The results are listed in Table \ref{Table 3}. In general, different values of $\alpha$ will give different optimal sequences. It is clear that both the single-run depth (depth of the optimal sequence) and the expected depth in Table \ref{Table 2} and \ref{Table 3} are smaller than that for Grover's algorithm (Table \ref{Table 1}). In practice, once $\alpha$ is known, one can guess the optimal sequence based on results with small $n$. For example, when $n$ is large, the optimal sequence is closed to (assuming that $n$ is even)
\begin{equation}
    S_{n,n/2}(1,1,2,\cdots,1,2,1,2) = G_{n/2}G_{n}G^2_{n/2}\cdots G_{n}G^2_{n/2}
\end{equation}
See Table \ref{Table 2}. The repetition number of $G_{n}G^2_{n/2}$ can be found either by numerical or analytical methods. 

\section{\label{Appendix alpha} Examples for Critical Ratios}

The ratio $\alpha$ defined in Eq. (\ref{def alpha}) is an important parameter. If $\alpha\rightarrow \infty$, Grover's algorithm is optimal in depth. The critical ratios $\alpha_{c,1}$ in Eq. (\ref{def alpha c}) and $\alpha_{c,2}$ in Eq. (\ref{def alpha c 2}) are threshold values. Below $\alpha_{c,1}$ (or $\alpha_{c,2}$), we can find a lower expected depth algorithm than Grover's algorithm. The diffusion operator $\text d(D_n)$ is single-qubit-gate equivalent to the $n$-qubit Toffoli gate $\Lambda_{n-1}(X)$. We can set $\text d(\Lambda_{n-1}(X))=\{1,5,13,29,61,120,160,200,240\}$ with $n=2,3,\ldots,10$, see Ref. \cite{BBCDMSSSW95}. Based on the depth optimization method defined in Secs. \ref{subsec: opt method} and \ref{subsec:divded and conquer}, we numerically find the critical ratios $\alpha_{c,1}$ and $\alpha_{c,2}$ in Table \ref{Table 4}.

\providecommand{\noopsort}[1]{}\providecommand{\singleletter}[1]{#1}%

\end{document}